\documentclass{article}
\usepackage[english]{babel}
\usepackage{graphicx}
\usepackage{amsmath}
\usepackage{tikz}
\usetikzlibrary{ shapes.geometric, arrows}
\usetikzlibrary{positioning}
\usetikzlibrary {arrows.meta,trees}
\usetikzlibrary{arrows}
\usepackage{amsthm}
\usepackage{amssymb}
\newtheorem{theorem}{Theorem}

\newtheorem{remark}{Remark}

\usepackage[ruled,vlined,linesnumbered,commentsnumbered]{algorithm2e}

\title{Top-k Stabbing Interval Queries}
\author{Waseem Akram {\normalsize{and}} Sanjeev Saxena\\
Computer Science and Engineering,\\
Indian Institute of Technology, Kanpur,\\
India 208016}
\date{5 November, 2024}

\begin{document}

\maketitle

\begin{abstract}
    We investigate a weighted variant of the interval stabbing problem, where the goal is to design an efficient data structure for a given set $\mathcal{I}$ of weighted intervals such that, for a query point $q$ and an integer $k>0$, we can report the $k$ intervals with largest weights among those stabbed by $q$. 
    In this paper, we present a linear space solution with $O(\log n + k)$ query time.
    Moreover, we also present another trade-off for the problem.
\end{abstract}

\section{Introduction}\label{sec:intro}
The \emph{interval stabbing problem} is an important problem in computational geometry \cite{deBerg08LSI,Amagata24,Chazelle86FS}, which asks to compute all intervals stabbed by a query point. In this paper, we study a natural variant of this problem where, instead of computing all stabbed intervals, the goal is to find a subset of the stabbed intervals that ``best'' represent them \cite{Amagata24}. Formally, the problem is defined as follows.
Let $\mathcal{I}$ be a set of $n$ intervals on the real line, each with a real valued weight. The set $\mathcal{I}$ is static.
The goal is to preprocess $\mathcal{I}$ so that, for a query value $q$ and a positive integer $k$, the $k$ intervals stabbed by $q$ with largest weights can be reported efficiently. An interval $[s_i, e_i]\in \mathcal{I}$ is stabbed by $q$ if and only if $s_i\le q \le e_i$. The problem is useful in cryptocurrency and stock applications \cite{Qiao16,Zhang15}.
We will refer this variant as \emph{the top-$k$ interval stabbing problem}.

In this paper, we present a linear space solution with $O(\log n + k)$ query time. In the rank space, the query time can be reduced to $O(k)$. Moreover, we also present a solution that uses only segment tree structure. It takes $O(n\log n)$ space and time to build, and $O(\log n + k \log\log n)$ time to answer a query. In rank space, the query time can be improved to $O(\log n + k)$.

Amagata et al. \cite{Amagata24} has proposed two algorithms for the problem; one algorithm uses interval tree structure while the other one employs segment tree structure. The algorithm based on interval tree structure takes linear space and $O(\sqrt{n}\log n + k)$ time to answer a query. The preprocessing takes $O(n\log n)$ time. The other algorithm based on segment tree achieves $O(\log n + k)$ query time, the preprocessing takes $O(n \log n \log\log n)$ time and $O(n\log^2 n)$ space. 

\section{Preliminaries}
\emph{The Orthogonal Line Segment Intersection Problem} is to preprocess a given set of horizontal segments in the plane so that queries of the following type can be answered efficiently: for a positive integer $k$ and a real value $x_0$, report the $k$ horizontal segments intersecting the vertical line $x=x_0$ in the top-to-down order \cite{Bille22}.
Using the \emph{hive-graph data structure} given by Chazelle \cite{Chazelle86FS}, one can solve the problem with $O(n)$ space and $O(\log n + k)$ query time. In the rank space where each coordinate and weight are integers in the range $[n]$, the query time can be reduced to $O(k)$ (see \cite[Lemma~$3$]{Bille22}).

We briefly describe \emph{the hive graph data structure} and explain how it can be used to answer orthogonal segment intersection queries. Let $H$ be a set of horizontal segments in the plane. The hive graph data structure for $H$ is a planar subdivision where each cell is a (possibly open) axis-parallel rectangle. Each vertical side in the subdivision passes through one of segments' endpoints. Then, the subdivision is preprocessed for \emph{the point location queries}, in which we are to compute the cell containing a query point.
To compute the horizontal segments intersected by a vertical segment $q$, we locate the cell that contains the upper endpoint of $q$, and then move along $q$ towards its lower endpoint. The horizontal segments encountered in the process are exactly those which are intersected by $q$. The segments are reported individually in the decreasing order of their $y$-coordinates. Note that if the cell containing the upper endpoint is known, finding the required horizontal segments takes time linear in the output size. See \cite[Section~4]{Chazelle86FS} for more details.

The \emph{segment tree}~\cite{deBerg08LSI} is a very useful data structure in computational geometry. It stores a set of $n$ intervals on the real line such that all interval stabbed by a query point $x_0$ can be reported in $O(\log n + \#output)$ time. It is a balanced binary search tree where the key of each node is a vertical slab $[a,b)\times \mathbb{R}$; the keys of the nodes of any particular level form a partition of the plane. The intervals 
are stored at nodes. An important property of the segment tree is that an interval is stored at a node of the root-to-leaf path for $x_0$ if and only if it is
stabbed by $x_0$.

\section{Optimal Solution}\label{sec:opt}
In this section, we 
transform the problem into another geometric problem, specifically, \emph{the orthogonal line segment intersection problem} \cite{Bille22}, for which there exists efficient algorithms.

A weighted interval can be interpreted as a horizontal segment in the plane, and these segments are preprocessed for segment intersection queries.
Let $\mathcal{I}=\{[s_i, e_i]| i\in [1, n]\}$ be the input set of intervals, where the weight of interval $[s_i, e_i]$ is $w_i$. For each interval $[s_i, e_i]$, we create a horizontal segment  $[s_i, e_i] \times w_i$: the $x$-coordinates are the endpoints of the interval and $y$-coordinate is the corresponding weight. We then construct a planar subdivision for these segments, as described by Chazelle in \cite{Chazelle86FS} (with vertical sides introduced by \emph{combing} technique). The data structure occupies $O(n)$ space.

\begin{figure}
    \centering
    \includegraphics[width=0.7\linewidth]{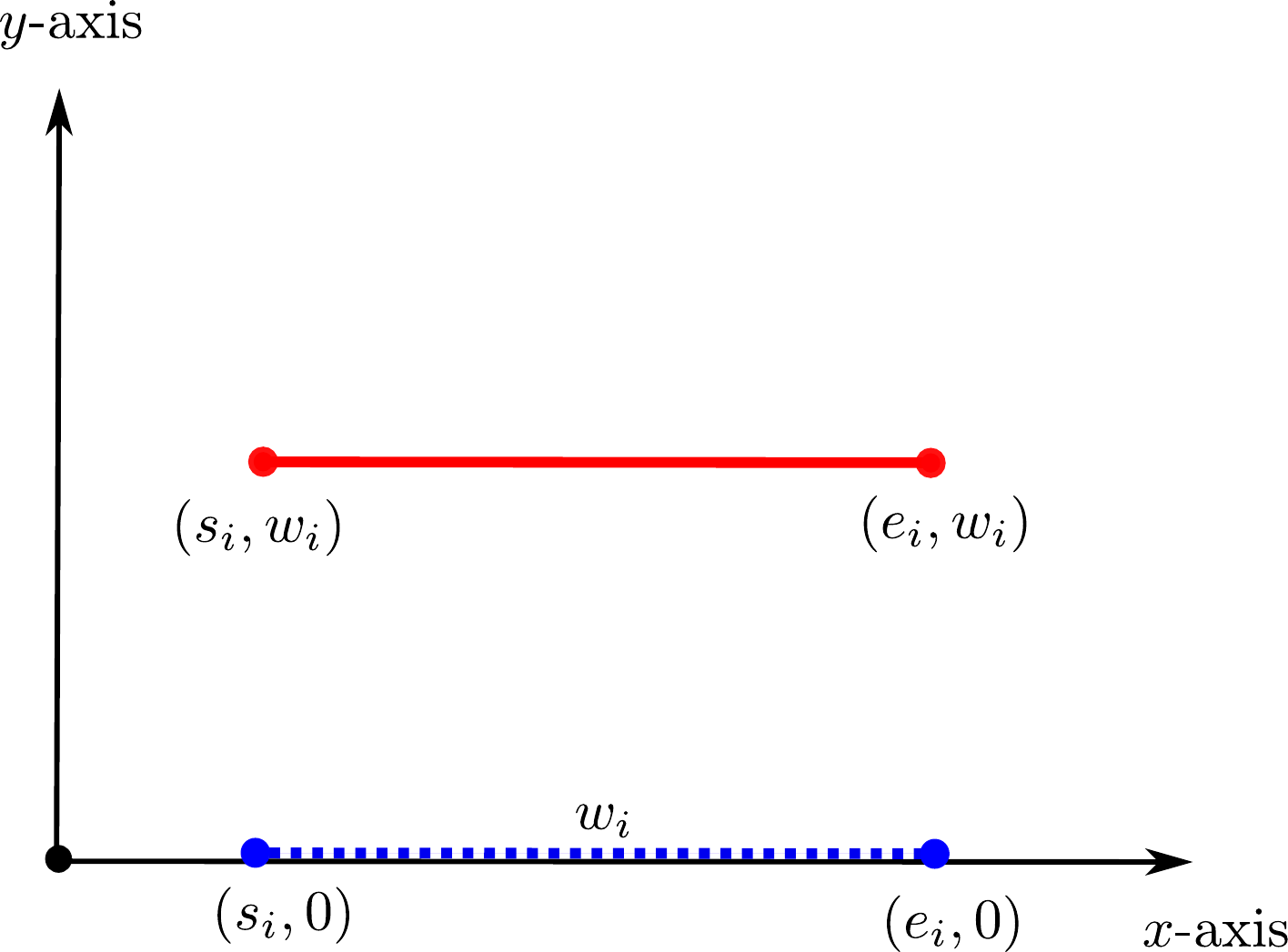}
    \caption{An interval $[s_i, e_i]$ with weight $w_i$ (dashed and blue) is interpreted as a horizontal line segment $[s_i, e_i]\times w_i$ in the plane (solid and red).}
    \label{fig:int-seg}
\end{figure}

For a given $1$-D point $x_0$ and a positive integer $k$, we locate the (open) cell in the subdivision containing the point $(x_0, +\infty)$ using a binary search. We then move downward along the vertical line $x=x_0$, and report the first $k$ segments (weighted intervals) intersected by the line. The query time is $O(\log n + k)$.

\begin{remark}
    If the intervals' endpoints and query points are all integers in the range $O(n)$, we can locate the cell enclosing the point $(x_0, +\infty)$ in constant time by maintaining a lookup table of size $O(n)$. As a consequence, the query time reduces to $O(k)$.
\end{remark}
\begin{remark}
    In the rank space, we can also employ the partial persistent data structure described in \cite[Lemma~$3$]{Bille22} and solve the problem using $O(n)$ space and $O(k)$ query time.
\end{remark}

\section{Segment Tree-Based Solution}\label{sec:segT}
In this section, we present an algorithm that employs only a segment tree and a heap structure. 
In the preprocessing phase, we build a segment tree $\mathcal{T}$ over the set $\mathcal{I}$ of weighted intervals. 
The canonical set $S(v)$ of each node $v\in \mathcal{T}$ is represented by an array, sorted by the intervals' weights.

We now describe how to answer a query.
Given a query point $x_0$ and a positive integer $k$, we first find the search path $\Pi$ for the point $x_0$ in the segment tree $\mathcal{T}$. We next access the first interval from every node $v\in\Pi$, and initialize a max heap $H$ with these intervals (with weights as keys). 
We keep repeating the following steps until $k$ intervals are reported or the heap $H$ becomes empty:
\begin{enumerate}
    \item Pop out the interval $x$ with maximal weight from $H$ and report it.
    \item Let $x$ belongs to node $v\in \Pi$. We find the next interval in the array stored at $v$ and insert it into the heap $H$.
\end{enumerate}
A pseudocode of the query algorithm is described in Algorithm~\ref{alg:segT}.
\begin{theorem}\label{thm:seg-tree}
    Given a set $\mathcal{I}$ of weighted intervals on the real line, we can represent $\mathcal{I}$ into a data structure that supports the top-$k$ interval stabbing queries in $O(\log n + k\log\log n)$ time. The data structure can be built in $O(n\log n)$ space and time.
\end{theorem}
\begin{proof}
    The correctness of the query algorithm is obvious. Finding the path $\Pi$ takes $O(\log n)$ time. Note that the initialized heap has elements no more than the length of the path $\Pi$.
    Since building a heap $H$ takes linear time, so $O(\log n)$ time is needed to initialize the heap $H$. For each reported interval, we spent $O(\log\log n)$ time. Thus, the total time to answer the query is $O(\log n + k\log\log n)$. 
    
    We first sort the set $\mathcal{I}$ of intervals in the decreasing order of their weights, and then construct the segment tree $\mathcal{T}$ over the sorted set. This ensures that the intervals at each node of $\mathcal{T}$ are sorted by their weights. Thus, the time to construct the data structure is $O(n\log n)$.
\end{proof}
\begin{remark}
    In the rank space, we can employ the \emph{fusion node structure} \cite{PatrascuT14} to implement heap structure, which takes $O(1)$ time for each operation. Thus, the query time reduces to $O(\log n + k)$.
\end{remark}
\begin{algorithm}[t]
    Find the search path $\Pi$ in the segment tree $\mathcal{T}$ for the value $x_0$\;
    $S\gets \emptyset$\;
    \ForEach{node $v\in \Pi$}{
        $z\gets$ the first interval present in the array of node $v$\;
        $S\gets S\cup \{z\}$\;
    }
    Build a max heap $H$ over the set $S$ 
    \Repeat{$k$ intervals are reported or $H$ becomes empty}{
        Pop out the interval $x$ with maximal weight from $H$, and report it\;
        \tcc{Let the interval $x$ came from the node $v\in \Pi$.}
        Insert the next interval from array of node $v$ into $H$\;
    }
    \caption{Segment Tree Based Algorithm}
    \label{alg:segT}
\end{algorithm}

\section{Conclusion}
We studied the top-$k$ interval stabbing query problem and presented efficient algorithms. First, we showed that problem can be transformed into a variant of \emph{the orthogonal segment intersection problem}. Consequently, we obtained the first linear time solution with $O(\log n + k)$ query time. 
We then presented a simple algorithm that employs only segment tree and heap structures, and solves the problem using $O(n\log n)$ space and $O(\log n + k\log \log n)$ query time.

\end{document}